\documentclass[journal,10pt]{IEEEtran}
\usepackage{amsmath,amssymb,amsthm}
\usepackage{graphicx}
\usepackage{cite}
\usepackage{subfigure}
\usepackage{flushend}

\newtheorem{proposition}{Proposition}
\newtheorem{corollary}{Corollary}
\newtheorem{remark}{Remark}
\usepackage{array,booktabs,tabularx}
\begin{document}

\title{Mobility Information Capacity in the Sky:\ A Gaussian Channel Perspective}

\author{Weijie~Yuan,~\IEEEmembership{Senior Member,~IEEE}, Fan Liu,~\IEEEmembership{Senior Member,~IEEE}, Shuangyang Li,~\IEEEmembership{Member,~IEEE}, Lin Zhou,~\IEEEmembership{Senior Member,~IEEE}, and Pingzhi Fan,~\IEEEmembership{Fellow,~IEEE} 
\thanks{W.~Yuan and L. Zhou are with the School of Automation and Intelligent Manufacturing, Southern University of Science and Technology, Shenzhen 518055, China (e-mail: yuanwj@sustech.edu.cn, zhoul9@sustech.edu.cn).

F. Liu is with the National Mobile Communications Research Laboratory, Southeast University, Nanjing 210096, China (e-mail: fan.liu@seu.edu.cn).

S. Li is with the Faculty of Electrical Engineering and Computer Science, Technical University of Berlin, 10587 Berlin, Germany (e-mail: shuangyang.li@tu-berlin.de).

P. Fan is with the Institute of Mobile Communications, Southwest Jiaotong University, Chengdu  611756, China (e-mail: pzfan@swjtu.edu.cn).
}}

\maketitle

\begin{abstract}
Existing airspace capacity metrics mainly quantify occupancy or flow, although the same number of aerial vehicles may result in different motion alternatives. This letter establishes \emph{mobility information capacity} as an information-theoretic measure for low-altitude wireless networks. It quantifies the maximum information that trajectory observations reveal about intentional maneuver inputs under a given maneuver-resource budget and environmental uncertainty. For a common fixed feedback architecture, we formulate a lifted linear-Gaussian mobility channel and derive its finite-horizon log-determinant capacity. Cost and uncertainty whitening gives the spatiotemporal mobility eigenmodes, whose optimal maneuver-resource allocation follows water-filling. %We then prove an operational coding theorem using independent repetitions of the finite-horizon channel.
When the number of nondegenerate modes grows linearly with time and their efficiencies become asymptotically symmetric, we arrive at the Shannon-like law $R_M^{\rm G}=\frac{B_M}{2}\log_2(1+\mathrm{MNR})$, where MNR is the mobility-to-noise ratio. %We also show that a bounded-dimensional terminal output has zero asymptotic rate under bounded modal gains.
The proposed measure opens a motion-centric capacity perspective for the sky, while remaining a distinguishability baseline rather than a collision- or geometry-constrained airspace capacity.
\end{abstract}

\begin{IEEEkeywords}
Airspace capacity, low-altitude wireless networks, mobility information capacity, water-filling.
\end{IEEEkeywords}

\section{Introduction}

Low-altitude wireless networks (LAWNs) are evolving from aerial communication infrastructures into platforms that jointly support sensing, communication, and control~\cite{yuan_lawn,jun2026low}. This evolution raises a fundamental question: how should the capacity of a LAWN be characterized? Existing studies primarily quantify airspace capacity in terms of usable volume, vehicle density, or traffic flow~\cite{cho_capacity,aarts_capacity,cummings_congestion}. However, the same number of drones occupying the same airspace may admit markedly different route choices and speed profiles. Occupancy alone therefore cannot capture the diversity of feasible motion. This motivates a broader notion of airspace capacity, i.e., once vehicles enter a region, how much freedom of motion remains available?

Shannon's operational viewpoint provides a natural starting point. Capacity counts the asymptotic growth of reliably distinguishable alternatives under specified resources and uncertainty~\cite{shannon,cover_thomas}. For a LAWN, one may similarly ask how rapidly motion alternatives can grow while their observed trajectories remain distinguishable. We treat the maneuver sequence as the input, the noisy trajectory as the output, maneuver effort as the resource, and wind, execution error, and sensing noise as uncertainty. This perspective complements occupancy- and flow-based notions of airspace capacity by asking an important motion-centric question: under a given operating configuration, \emph{how much distinguishable motion can the airspace support?} 

Building on action-to-state information and empowerment~\cite{klyubin_empowerment,touchette_control,tiomkin_control,tiomkin_intrinsic}, we define mobility information capacity (MIC) as the maximum input--trajectory mutual information under a maneuver-resource budget. Table~\ref{tab:capacity_metrics} summarizes its relation to conventional airspace metrics. The proposed MIC can be used to compare sensing architectures, maneuver efficiency, and information loss from observed trajectory under a common resource budget. It does not directly determine safe aircraft counts or traffic throughput.

\begin{table}[t]
    \caption{Comparison of Conventional Airspace Metrics and MIC}
    \label{tab:capacity_metrics}
    \centering
    \footnotesize
    \setlength{\tabcolsep}{3pt}
    \renewcommand{\arraystretch}{1.15}
    \begin{tabularx}{\columnwidth}{
        @{}
        >{\raggedright\arraybackslash}p{0.21\columnwidth}
        >{\raggedright\arraybackslash}p{0.24\columnwidth}
        >{\raggedright\arraybackslash}X
        @{}
    }
        \toprule
        \textbf{Metric}
        & \textbf{Typical units}
        & \textbf{Meaning } \\
        \midrule

        Occupancy
        & Number of aircrafts 
        & Spatial loading  \\

        \addlinespace[2pt]
        Flow
        & Aircraft per unit time
        & Traffic throughput\\

        \addlinespace[2pt]
       MIC
        & Bits over $T$
        & Observable maneuver diversity under specified
          resources and uncertainty \\

        \bottomrule
    \end{tabularx}
\end{table}

The main contributions of the letter are threefold. First, we formulate a finite-horizon lifted mobility channel and define its capacity-cost function. For linear-Gaussian dynamics, we derive an exact log-determinant expression for the resulting MIC. Second, by whitening maneuver cost and trajectory uncertainty, we show that the optimal maneuver-resource allocation follows water-filling. This representation reveals how dynamics, sensing quality, uncertainty, and maneuver cost jointly determine the observable motion. Third, we further derive a Shannon-like information rate under linear growth of nondegenerate modes and asymptotically equal gains, and show that fixed-dimensional terminal observations have zero asymptotic rate under bounded mobility gains.

\textit{Notations:} $(\cdot)^{\mathsf T}$ denotes transpose; $\mathbf I$ denotes an identity
matrix of appropriate dimension; $\operatorname{tr}(\cdot)$,
$\det(\cdot)$, and $\operatorname{rank}(\cdot)$ denote the trace,
determinant, and rank, respectively. For symmetric matrices,
$\mathbf A\succeq\mathbf 0$ ($\mathbf A\succ\mathbf 0$) denotes positive
semidefiniteness (positive definiteness).
$\mathbb E[\cdot]$ denotes expectation,
$\operatorname{supp}(\cdot)$ denotes the support of a distribution,
and $[x]^+\triangleq\max\{x,0\}$.
The notation $\mathbf x\perp\mathbf y$ denotes statistical independence,
$\mathcal N(\mathbf m,\mathbf R)$ denotes a Gaussian distribution with
mean $\mathbf m$ and covariance $\mathbf R$.
All logarithms are to base two unless otherwise stated.

%The operational chain is maneuver label $\rightarrow$ input sequence $\rightarrow$ fixed closed-loop dynamics $\rightarrow$ noisy trajectory $\rightarrow$ identification. Although the transformation is dynamical rather than electromagnetic, the communication question remains: how many alternatives can a noisy input--output mechanism reliably separate?

\section{Mobility Channel and Information Capacity}
%Two maneuver sequences may lead to the same terminal state while producing different intermediate trajectories. How much information about the selected maneuver can the available observations reveal?

Consider a vehicle operating in a LAWN under a fixed sensing–communication–control configuration\cite{jin_codesign}. A mission-level maneuver sequence is executed through a common feedback controller, while onboard and infrastructure-assisted sensing produces the trajectory observation available to the network. The feedback controller is fixed across all mobility alternatives, which differ only through the exogenous maneuver sequence. Over an observation duration $T=K\Delta t$, let $\mathbf U=[\mathbf u_0^{\mathsf T},\ldots,
\mathbf u_{K-1}^{\mathsf T}]^{\mathsf T}$ denote the exogenous maneuver input sequence and $\mathbf Y=[\mathbf y_1^{\mathsf T},\ldots,
\mathbf y_K^{\mathsf T}]^{\mathsf T}$ denote the trajectory observations. For a known initial state and known system matrices, the
deterministic initial-state response is common to all maneuver
sequences and can be subtracted. By using $\mathbf Y$ as the resulting observations, we obtain a real-valued vector Gaussian channel:
\begin{equation}
    \mathbf Y=\mathbf G_K\mathbf U+\mathbf Z_K,
    \label{eq:channel}
\end{equation}
where $\mathbf G_K$ denotes the maneuver-to-trajectory mapping matrix,
$\mathbf{Z}_K\sim\mathcal N(\mathbf{0},\mathbf{R}_{Z,K})$ with $\mathbf{R}_{Z,K}\succ\mathbf{0}$ includes propagated process, execution, localization, and observation errors. Fig.~\ref{fig:physical-to-channel} summarizes how the physical mobility system induces the finite-horizon Gaussian channel in \eqref{eq:channel}. We further assume that $\mathbf{U}\perp\mathbf{Z}_K$, which follows the fact that the maneuver sequence is selected exogenously, prior to the realization of execution and observation uncertainties, and is statistically independent of these disturbances\footnote{If the maneuver input is instead generated by a causal policy that adapts to past observations, $\mathbf{U}$ and $\mathbf{Z}_K$ are generally no longer independent, and a causal information measure such as directed information may be required.}.

\begin{figure}[t]
    \centering
    \includegraphics[width=1.05\columnwidth]
    {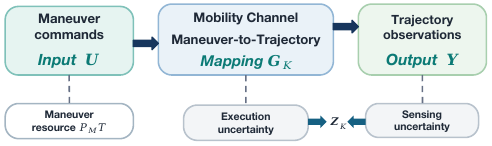}
    \caption{Finite-horizon mobility channel.}
    \label{fig:physical-to-channel}
\end{figure}

The channel matrix $\mathbf G_K$ is determined from the finite-horizon linear-Gaussian model
\begin{align}
    \mathbf x_{k+1}
    &=\mathbf A_k\mathbf x_k+\mathbf B_k\mathbf u_k
      +\mathbf w_k,\\
    \mathbf y_{k+1}
    &=\mathbf C_{k+1}\mathbf x_{k+1}+\mathbf v_{k+1},
    \label{eq:local}
\end{align}
where $\mathbf x_k$ is the mobility state, $\mathbf{u}_k$ is a mission-dependent maneuver command, $\mathbf{y}_k$ is the corresponding observation, and $\mathbf w_k$ and $\mathbf v_k$ represent process/execution and observation errors, respectively. The matrix $\mathbf A_k$ incorporates the common feedback
controller, e.g., stabilizing feedback, $\mathbf B_k$ maps commands to state increments, and $\mathbf C_{k+1}$ specifies the observed state components. For example, for a position-velocity state, position-only sensing uses
$\mathbf C_{k+1}=[\mathbf I\ \mathbf 0]$. Stacking these relations gives \eqref{eq:channel}, where the $(k,j)$ block of $\mathbf{G}_K$ is $\mathbf{C}_k\mathbf{H}(k,j+1)\mathbf{B}_j$ for $j<k$ with\footnote{After stacking (2) and (3), there exists a term \(\mathbf F_K\mathbf x_0\) related to the initial state \(\mathbf x_0\). Since both the transition matrix $F_K$ and state \(\mathbf x_0\) are known, subtracting this term leaves mutual information unchanged. Motion differences induced by the candidate commands remain in \(\mathbf G_K\mathbf U\).} 
\begin{equation}
\mathbf{H}(k,j+1)=\prod_{i=1}^{k-j-1} \mathbf{A}_{k-i},
\end{equation}
and $\mathbf{H}(k,k)=\mathbf{I}$.

%Alternatives change this request, not the stabilizing controller.

% For the scalar two-slot model $x_{k+1}=ax_k+u_k+w_k$, $y_k=x_k+v_k$, and $x_0=0$,
% \begin{equation}
%  \begin{bmatrix}y_1\\y_2\end{bmatrix}
%  =\begin{bmatrix}1&0\\a&1\end{bmatrix}
%  \begin{bmatrix}u_0\\u_1\end{bmatrix}
%  +\begin{bmatrix}z_1\\z_2\end{bmatrix}.                                      \label{eq:twoslot}
% \end{equation}
% The first action affects both observations; the last affects only $y_2$. This causal propagation gives $\mathbf{G}_K$ a lower block-triangular structure.

% \begin{remark}
% Although $\mathbf{Z}_K$ aggregates disturbances over the entire horizon, its per-slot uncertainty need not grow unboundedly with $K$. For the stable closed-loop tracking systems, the influence of an earlier disturbance decays over time. Under stationary or uniformly bounded process and observation noise, the state-error covariance therefore remains bounded and, in the time-invariant case, typically converges to a steady-state covariance. Hence, increasing the horizon enlarges the dimension of $\mathbf{R}_{Z,K}$ rather than increasing the uncertainty.
% \end{remark}

To quantify the maneuver effort required by different mobility alternatives, we introduce a positive-definite weighting matrix $\mathbf{Q}_K\succ\mathbf{0}$ and impose the quadratic resource constraint
\begin{equation}
    \mathbb E\!\left[
    \mathbf{U}^{\mathsf T}\mathbf{Q}_K\mathbf{U}
    \right]
    \leq P_M T ,
    \label{eq:cost}
\end{equation}
where $\mathbf{Q}_K$ defines the relative cost of different mobility directions. The term $P_M$ is a maneuver-resource rate and $P_M T$ denotes the total maneuver-resource budget over the horizon. In particular, higher weights in $\mathbf{Q}_K$ penalize mobility directions that are more costly to execute. Depending on the choice of $\mathbf{u}_k$ and
$\mathbf{Q}_K$, the quadratic cost may represent control energy,
acceleration effort, or trajectory-deviation effort. 

Given the Gaussian mobility channel \eqref{eq:channel} and resource constraint \eqref{eq:cost}, the finite-horizon MIC is defined as
\begin{equation}
 C_M^{\rm G}(T)\triangleq
 \sup_{p_\mathbf{U}:\,\mathbb E[\mathbf{U}^{\mathsf T}\mathbf{Q}_K\mathbf{U}]\leq P_M T} I(\mathbf{U};\mathbf{Y}). \label{eq:capacitydef}
\end{equation}
Mutual information is used to quantify the trajectory variation that is attributable to the intentional maneuver input. The quantity \(C_M^{\rm G}(T)\) is measured in bits over the observation horizon \(T\). Its time-normalized value \(C_M^{\rm G}(T)/T\) is measured in bits per unit time. Here, bits quantify the information that trajectory observations reveal about the intentional maneuver input. The metric is an information-theoretic measure of observable maneuver diversity, rather than a count of vehicles or a count of trajectories. In particular, stronger disturbances may increase the trajectory entropy $h(\mathbf{Y})$ without creating additional distinguishable maneuver alternatives. Under $\mathbf{U}\perp\mathbf{Z}_K$, 
\begin{align}
I(\mathbf{U};\mathbf{Y})
&=h(\mathbf{Y})-h(\mathbf{Y}\mid\mathbf{U})\nonumber\\
&=h(\mathbf{Y})-h(\mathbf{Z}_K),
\end{align}
which excludes the uncertainty-induced spread from the useful mobility
information.

\begin{proposition}[Gaussian mobility information capacity]
For~\eqref{eq:channel}--\eqref{eq:cost},
\begin{equation}
 \begin{aligned}
 C_M^{\rm G}(T)=\max_{\substack{\mathbf{K}_U\succeq\mathbf{0}\\
 \operatorname{tr}(\mathbf{Q}_K\mathbf{K}_U)\leq P_MT}}\;&\frac{1}{2}\log_2\det\!\bigl(\mathbf{I}+\\[-1mm]
 &\mathbf{R}_{Z,K}^{-1/2}\mathbf{G}_K\mathbf{K}_U\mathbf{G}_K^{\mathsf T}
 \mathbf{R}_{Z,K}^{-1/2}\bigr),
 \end{aligned}                                                                      \label{eq:logdet}
\end{equation}
where $\mathbf{K}_U$ denotes the covariance matrix of $\mathbf{U}$.
\end{proposition}

\begin{IEEEproof}
Let $\mathbf{U}$ have mean $\mathbf{m}_U$ and covariance $\mathbf{K}_U$. Its average cost is $\operatorname{tr}(\mathbf{Q}_K\mathbf{K}_U)
+\mathbf{m}_U^{\mathsf T}\mathbf{Q}_K\mathbf{m}_U$.
Centering $\mathbf{U}$ preserves mutual information while eliminating
the second cost term, so an optimal input may be taken to be zero mean.
For fixed $\mathbf{K}_U$, Gaussian extremality maximizes
$h(\mathbf{Y})$ for a Gaussian output with covariance
$\mathbf{G}_K\mathbf{K}_U\mathbf{G}_K^{\mathsf T}
+\mathbf{R}_{Z,K}$ and the bound is attained
by a zero-mean Gaussian input. Thus, it remains to optimize the covariance $\mathbf{K}_U$ under maneuver-resource constraint, which yields \eqref{eq:logdet}. The factor \(1/2\) follows from the real-valued Gaussian model.
\end{IEEEproof}

The covariance optimization in \eqref{eq:logdet} distributes the
available maneuver resource across observable mobility directions. After
uncertainty whitening, the determinant measures the multiplicative
expansion of trajectory variation induced by the intentional maneuver
input.

\section{Mobility Mode and Shannon-Like Law}

The MIC defined above does not yet reveal which kinds of motion are valuable. To do so, we define the cost- and uncertainty-whitened mapping
\begin{equation}
\widetilde{\mathbf{G}}_K=\mathbf{R}_{Z,K}^{-1/2}\mathbf{G}_K\mathbf{Q}_K^{-1/2}.    \label{eq:whiten}
\end{equation}
Defining $\widetilde{\mathbf U}=\mathbf Q_K^{1/2}\mathbf U$ converts the resource constraint into $\mathbb E[\|\widetilde{\mathbf U}\|^2]\leq P_MT$. Equal-length vectors in these coordinates therefore have equal maneuver cost. For a two-dimensional example with
$\mathbf Q=\operatorname{diag}(1,4)$, the commands $(1,0)^{\mathsf T}$ and $(0,1/2)^{\mathsf T}$ both have unit cost and become unit vectors after the transformation. The uncertainty-whitening factor $\mathbf R_{Z,K}^{-1/2}$ decorrelates the trajectory error and scales its covariance to the identity matrix. Thus, $\widetilde{\mathbf{G}}_K$ maps one unit of normalized maneuver resource into noise-normalized trajectory variation. Let $\lambda_1,\ldots,\lambda_{L(T)}>0$ be the nonzero eigenvalues of $\widetilde{\mathbf{G}}_K^{\mathsf T}\widetilde{\mathbf{G}}_K$, where $L(T)=\operatorname{rank}(\mathbf{G}_K)$.
\begin{proposition}[Mobility mode decomposition]
\begin{equation}
 C_M^{\rm G}(T)=\max_{\substack{p_\ell\geq0\\\sum_{\ell=1}^{L(T)}p_\ell\leq P_MT}}
 \frac{1}{2}\sum_{\ell=1}^{L(T)}\log_2(1+\lambda_\ell p_\ell),                  \label{eq:parallel}
\end{equation}
The optimal resource allocation is
\begin{equation}
 p_\ell^\star
 =
 \left(\nu-\frac{1}{\lambda_\ell}\right)^+,
 \qquad
 \sum_{\ell=1}^{L(T)}p_\ell^\star=P_MT,
 \label{eq:waterfill}
\end{equation}
where $\nu$ is the water level.
\end{proposition}
\begin{IEEEproof}
Define
$\widetilde{\mathbf K}_U
=\mathbf Q_K^{1/2}\mathbf K_U\mathbf Q_K^{1/2}$,
so that
$\operatorname{tr}(\widetilde{\mathbf K}_U)\leq P_MT$.
Using the singular-value decomposition
$\widetilde{\mathbf G}_K
=\mathbf H\mathbf\Sigma\mathbf V^{\mathsf T}$ and Sylvester's
determinant identity, \eqref{eq:logdet} can be expressed in the
eigenmode coordinates of $\widetilde{\mathbf G}_K$. Hadamard's inequality shows that the optimum is attained when
$\mathbf V^{\mathsf T}\widetilde{\mathbf K}_U\mathbf V$ is diagonal.
Denoting its diagonal entries by $p_\ell$ yields
\eqref{eq:parallel}, whose KKT conditions give the water-filling
solution \eqref{eq:waterfill}. Importantly, the eigenmode decomposition
depends on the known finite-horizon model and uncertainty
statistics. Hence, the water-filling solution is an offline design before the mission.
\end{IEEEproof}

Let $\mathbf v_\ell$ be a unit eigenvector of $\widetilde{\mathbf G}_K^{\mathsf T}\widetilde{\mathbf G}_K$ associated with $\lambda_\ell$. It represents a maneuver mode in the cost-whitened coordinates, with the corresponding command pattern in the original
coordinates given by $\mathbf Q_K^{-1/2}\mathbf v_\ell$. Hence, $\lambda_\ell$ measures the observable trajectory response generated by one unit of normalized maneuver resource, while
$p_\ell$ specifies how much resource is allocated to that maneuver pattern. Their product $\lambda_\ell p_\ell$ is therefore the effective mobility-to-noise ratio (MNR) of the $\ell$th mode. For example, $\lambda_\ell p_\ell=1$ means that the intentional trajectory variation along this mode has the same variance as the normalized uncertainty, whereas a much larger value indicates that the maneuver pattern can be distinguished more reliably. A zero eigenvalue corresponds to a maneuver pattern that produces no observable signature in the selected measurements. Accordingly, $\lambda_\ell$ jointly captures the dynamics and observation mapping through $\mathbf{G}_K$, maneuver cost through $\mathbf{Q}_K$, and uncertainty through $\mathbf{R}_{Z,K}$. Equation~\eqref{eq:waterfill} therefore allocates more resource to modes with larger observable motion per unit cost and leaves sufficiently weak modes inactive.

The finite-horizon MIC in \eqref{eq:parallel} depends on the full modal spectrum. A compact scalar law follows under the following sufficient conditions: useful modes have a linear growth rate $L(T)=B_MT+o(T)$ with $B_M>0$ and their non-zero mobility gains converge to a common value, i.e., $\lambda_{\min}(T)\to\lambda$, $\lambda_{\max}(T)\to\lambda$, $\lambda>0$, where $\lambda_{\min}(T)$ and $\lambda_{\max}(T)$ are the extrema of the nonzero gains\footnote{{Equal modal gains can arise when maneuver inputs and temporally white process noise pass through the same dynamics, with exact full-state sensing and uniform maneuver costs. For example, consider $x_{k+1}=ax_k+u_k+w_k$ with a known initial state, full-state observations $y_k=x_k$, and $\mathbf Q_K=\mathbf I$.
Let $w_k$ be independent Gaussian process noise with variance
$\sigma_w^2>0$. We can see that all mobility gains equal $\sigma_w^{-2}$ for every horizon, which satisfies the asymptotic equal-gain condition.}}. Thus, approximately $B_MT$ independent real modes are available through a $T$-horizon, and $B_M$ is termed the \emph{useful-mode creation rate}, which reflects how rapidly the control–dynamics–observation system reveals new independent spatiotemporal mobility directions as the observation horizon is extended\footnote{Linear growth of $L(T)$ is typical when each additional control interval introduces new maneuver commands whose effects are retained in the observed trajectory. For example, for an input $\mathbf{u}_k\in\mathbb{R}^{d_u}$ with double-integrator motion, each acceleration command $\mathbf u_k$ changes the next observed position by $(\Delta t^2/2)\mathbf u_k$. Under full-trajectory observation, each time step therefore contributes $d_u$ new independent mobility directions, giving $L(T)=d_u T/\Delta t$.}.

\begin{corollary}[Shannon-like long-horizon mobility information rate]
\begin{equation}
 R_M^{\rm G}
 \triangleq
 \lim_{T\to\infty}\frac{C_M^{\rm G}(T)}{T}
 =
 \frac{B_M}{2}\log_2(1+\mathrm{MNR}),
 \label{eq:shannonlike}
\end{equation}
where $\mathrm{MNR} \triangleq \frac{\lambda P_M}{B_M}$, $\frac{C_M^{\rm G}(T)}{T}$ denotes the finite-horizon mobility information rate.
\end{corollary}

\begin{IEEEproof}
Equal allocation $p_\ell=P_MT/L(T)$ is feasible in
\eqref{eq:parallel}. Since
$\lambda_\ell(T)\geq\lambda_{\min}(T)$, it gives
\[
 \frac{C_M^{\rm G}(T)}{T}
 \geq
 \frac{L(T)}{2T}
 \log_2\!\left(
 1+\lambda_{\min}(T)\frac{P_MT}{L(T)}
 \right).
\]
Conversely, since
$\lambda_\ell(T)\leq\lambda_{\max}(T)$, replacing all mobility gains by
$\lambda_{\max}(T)$ can only increase the capacity. Jensen's
inequality together with $\sum_\ell p_\ell\leq P_MT$ then gives
\[
 \frac{C_M^{\rm G}(T)}{T}
 \leq
 \frac{L(T)}{2T}
 \log_2\!\left(
 1+\lambda_{\max}(T)\frac{P_MT}{L(T)}
 \right).
\]
Under $\lambda_{\min}(T)\to\lambda$ and $\lambda_{\max}(T)\to\lambda$, both bounds converge to
$({B_M}/{2})\log_2(1+\lambda P_M/B_M)$, proving
\eqref{eq:shannonlike}.
\end{IEEEproof}

Here, $P_M/(B_M)$ is the asymptotic maneuver resource available per mobility mode, while $\lambda$ measures the observable trajectory response per unit normalized maneuver resource. This limit gives the maximum mutual information per unit physical time in the asymptotically equal-gain regime.

% For a family of asymptotically symmetric-mode systems with fixed
% $\lambda$ and $P_M$, we have 
% \begin{equation}
%  R_M^{\rm G}(B_M) \to \frac{\lambda P_M}{2\ln 2},
%  \qquad \textrm{when}~B_M\to\infty.
%  \label{eq:wideband}
% \end{equation}
% Thus, increasing $B_M$ raises the mobility-information rate, but with diminishing returns. Over a horizon $T$, approximately $B_MT$ modes share the total maneuver resource $P_MT$. The resource available per mode, i.e., $P_M/(B_M)$, decreases as more modes become available. This relation identifies three ways to improve $R_M^{\rm G}$. First, $B_M$ can be increased by introducing new controllable and observable maneuver directions, for example through additional independent actuation that reveal previously unobservable motion. Merely sampling the same physical trajectory more densely, without adding independent input-output directions, does not increase $B_M$. Second, stronger maneuver authority can increase the resource rate $P_M$. Third, improved system responsiveness, maneuver efficiency, disturbance rejection, or sensing quality can increase the whitened modal gains $\lambda_\ell$. The first mechanism creates or reveals additional modes, whereas the latter two provide more resource to or improve the efficiency of the existing modes.

Finally, we would like to emphasize that the extensive-mode condition is not automatic. If only a terminal output of fixed dimension is observed, $L(T)\leq L_{\max}$, and $\lambda_\ell(T)\leq\lambda_{\max}<\infty$, then
\begin{equation}
 {C_M^{\rm G}(T)}\leq\frac{L_{\max}}{2}
 \log_2(1+\frac{\lambda_{\max}P_MT}{L_{\max}}),                                   \label{eq:terminal}
\end{equation}
which only grows at the order of $\log T$, resulting in a vanishing information rate when $T\to\infty$. This is because a fixed-dimensional terminal observation retains only the aggregate effect of the maneuver sequence on the endpoint. Consequently, various trajectories departing in different directions but reach the same terminal state cannot be distinguished. In contrast, full-trajectory sensing can preserve an extensive collection of independent temporal maneuver patterns. Therefore, under bounded mobility gains, a positive long-horizon mobility-information rate requires the number of useful modes to grow linearly with time.
\begin{remark}[Intermediate observation architectures]
Full-trajectory recording is not necessary for a positive long-horizon information rate. Periodic sampling can also support a positive rate if it retains a number of modes proportional to the horizon whose gains remain uniformly bounded away from zero. For sliding-window processing, the long-horizon rate also depends on whether information extracted from earlier windows is preserved.
\end{remark}

\section{Numerical Results}
\label{sec:numerical}
We consider the following dynamic model
\begin{equation}
 x_{k+1}=0.85x_k+u_k+w_k,\qquad
 y_k=x_k+v_k.
 \label{eq:scalar_sim}
\end{equation}
which represent one lateral deviation coordinate of a UAV under a
stabilizing controller. The factor 0.85 describes the decay of the mean deviation in the absence of new maneuver inputs.
We set $x_0=0$, $\Delta t=1$, $\sigma_w^2=0.04$, $\sigma_v^2=0.01$, $\mathbf{Q}_K=\mathbf{I}$, and $P_M=0.01$. Here, $w_k$ represents wind or execution uncertainty, whereas $v_k$ represents localization or observation noise. This example illustrates the general unequal-gain case, whose capacity is evaluated using the exact water-filling
solution in \eqref{eq:parallel}.

\begin{figure}[t]
    \centering
    \includegraphics[width=0.95\columnwidth]
    {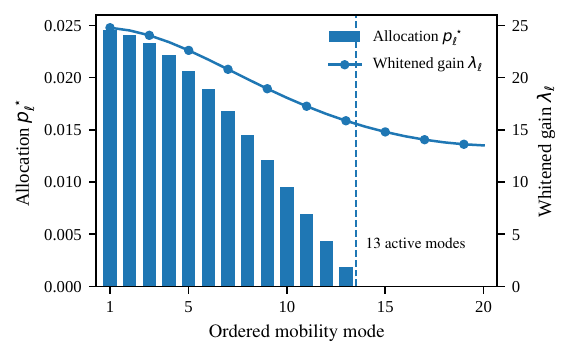}
    \caption{Finite-horizon mobility eigenmodes and water-filling
    allocation for \eqref{eq:scalar_sim}.}
    \label{fig:mode_waterfill}
\end{figure}
Fig.~\ref{fig:mode_waterfill} illustrates the finite-horizon result in
Proposition~2. It can be seen that even for the time-invariant system in
\eqref{eq:scalar_sim}, the mobility modes have different gains, because an earlier maneuver input can influence multiple subsequent
observations, whereas a later input has fewer opportunities to appear
in the observed trajectory. Consequently, water-filling allocates no
resource to seven sufficiently weak modes at $K=20$.%These modes represent spatiotemporal maneuver patterns over the entire horizon, rather than individual time slots.
%This example also shows that the exact finite-horizon capacity is given by \eqref{eq:parallel} instead of \eqref{eq:shannonlike} while requires the 

% Next, we illustrate the additional asymptotic conditions of Corollary~1 required by Shannon-like expression in \eqref{eq:shannonlike}. We consider an
% asymptotically homogeneous modal benchmark with
% $L(T)=T$, corresponding to $B_M=1/2$, and $\lambda_\ell(T)
%  =
%  \lambda\!\left[
%  1+\frac{0.8}{\sqrt{T}}
%  \cos\!\left(
%  \frac{2\pi(\ell-\frac{1}{2})}{L(T)}
%  \right)
%  \right],~\lambda=20$.
% We can observe that both $\lambda_{\min}(T)$ and
% $\lambda_{\max}(T)$ converges to $\lambda$ as $T\to\infty$.
% With $P_M=0.01$, Fig.~\ref{fig:shannon_convergence} shows that the
% exact water-filling rate remains between the lower and upper bounds
% used in the proof of Corollary~1. As the modal variance decreases,
% the exact mobility information capacity converges to
% \begin{equation}
%  B_M\log_2\!\left(
%  1+\frac{\lambda P_M}{B_M}
%  \right)
%  \simeq 0.132~\text{bits/slot},
% \end{equation}
% thereby numerically illustrating the Shannon-like mobility law.

\begin{figure}[t]
    \centering
    \includegraphics[width=0.95\columnwidth]
    {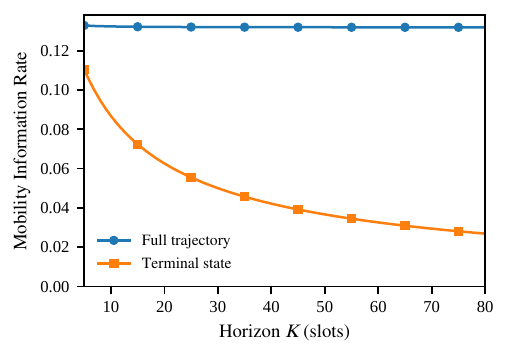}
    \caption{Finite-horizon mobility-information rate $C_M^G(T)/T$
    under full-trajectory and terminal-only observations for
    \eqref{eq:scalar_sim}.}
    \label{fig:trajectory_terminal}
\end{figure}

Finally, Fig.~\ref{fig:trajectory_terminal} compares two observation
architectures for the same physical model \eqref{eq:scalar_sim}. The
full-trajectory observer retains $\mathbf{Y}=[y_1,\ldots,y_K]^{\mathsf T}$, whereas the terminal observer retains only $y_K$. It can be seen that the full-trajectory observer preserves an increasing number of modes and maintains an approximately constant finite-horizon rate. In contrast, the mobility information rate corresponding to terminal state decays with $K$, consistent with \eqref{eq:terminal}. At $K=80$, the full-trajectory and terminal-only rates are approximately $0.132$ and $0.027$ bits/slot, respectively. This distinction validates that increasing the maneuver horizon alone does not create a positive long-horizon mobility-information rate $R_M^G$. The observation architecture must also preserve a growing number of informative motion modes.

\section{Extensions and Conclusion}
\label{sec:conclusion}
The proposed mobility information capacity quantifies the distinguishability of intentional maneuver alternatives rather than their physical admissibility. In particular, the capacity-achieving Gaussian input has unbounded support and may violate actuator, fly-zone boundary, or collision constraints. Let $\mathcal{E}_T$ denote a horizon-wise violation event and let $\mathcal{U}_{\epsilon}$ collect the maneuver sequences satisfying
$\Pr(\mathcal{E}_T\mid\mathbf{U}=\mathbf{u})\leq\epsilon$. A corresponding
safety-constrained capacity is
\begin{equation}
 C_{M,\epsilon}(T)
 \triangleq
 \sup_{\substack{
 p(\mathbf{U}):\,\operatorname{supp}(p)\subseteq\mathcal{U}_{\epsilon}\\
 \mathbb E[\mathbf{U}^{\mathsf T}\mathbf{Q}_K\mathbf{U}]\leq P_MT}}
 I(\mathbf{U};\mathbf{Y})
 \leq C_M^{\rm G}(T).
 \label{eq:safety_capacity}
\end{equation}
Under these admissibility constraints, the Gaussian input may no longer be feasible, and the log-determinant expression generally serves only as an upper bound.

The present formulation also assumes a common feedback controller and
an exogenously selected maneuver sequence. If the maneuver input adapts
causally to past observations, e.g., $\mathbf{u}_k=f_k(\mathbf{y}_{0:k})$, ordinary mutual information may need to be replaced by a causal measure such as directed information~\cite{massey_directed}. Likewise, multi-vehicle collision constraints couple feasible maneuver sequences across agents. Extending mobility information capacity to such causal and multi-agent settings is an important direction for future work.

This letter introduced mobility information capacity as a motion-centric information measure for low-altitude wireless networks. A lifted linear-Gaussian model yields a finite-horizon log-determinant capacity, and whitening maneuver cost and observation uncertainty gives mobility eigenmodes with a water-filling resource allocation. Under linear mode growth and asymptotically equal gains, the normalized capacity takes the Shannon-like form \(R_M^{\rm G}=\frac{B_M}{2}\log_2(1+\mathrm{MNR})\). The observation architecture determines which maneuver modes remain informative as the horizon grows. In particular, fixed-dimensional terminal observations have a vanishing normalized rate under uniformly bounded mobility gains. These results provide a baseline for comparing observable maneuver diversity under specified dynamics, resource budgets, and uncertainty. Incorporating physical admissibility and multi-vehicle interactions remains necessary for a safety-constrained airspace-capacity measure.

\bibliographystyle{IEEEtran}
\bibliography{space_time_MIC_WCL}
\end{document}